\documentclass[letterpaper, 10 pt, conference]{ieeeconf}  %
\IEEEoverridecommandlockouts                              %
\usepackage{graphicx,subcaption}
\usepackage{amsmath,amssymb}
\usepackage{mathrsfs}
\usepackage{booktabs}
\usepackage{color}
\usepackage{mathrsfs}
\usepackage{algorithm}

\usepackage{tikz}
\usetikzlibrary{arrows.meta}
\usetikzlibrary{positioning}
\usetikzlibrary{fit}
\usetikzlibrary{calc}

\newtheorem{theorem}{Theorem}
\newtheorem{assumption}{Assumption}

\newtheorem{proposition}{Proposition}
\newtheorem{lemma}{Lemma}
\newtheorem{remark}{Remark}

\newcommand{\Ltwo}{\boldsymbol{\rm L}_{2}}
 
\newcommand{\Ltwoe}{\boldsymbol{\rm L}_{2e}} 

\usepackage{amssymb}

\title{\LARGE \bf
Distributed Synchronisation of Heterogeneous Dynamical Networks With Nonlinear Diffusive Couplings
}
\author{Yongkang Su, Joaquin Carrasco, I\~{n}aki Esnaola and Lanlan Su
\thanks{Y. Su and I. Esnaola are with the School of Electrical and Electronic Engineering, University of Sheffield, S1 3JD, Sheffield, UK
        {\small (e-mail: ysu34@sheffield.ac.uk; esnaola@sheffield.ac.uk)}}%
\thanks{J. Carrasco and L. Su are with the Department of Electrical and Electronic Engineering, University of Manchester, M13 9PL, Manchester, UK
        {\small (e-mail: joaquin.carrasco@manchester.ac.uk; lanlan.su@manchester.ac.uk)}}%
}

\begin{document}
\maketitle
\thispagestyle{empty}
\pagestyle{empty}

\begin{abstract}
This letter investigates the problem of output synchronisation in heterogeneous dynamical networks with nonlinear diffusive couplings in the presence of disturbances on the coupling links. By exploiting relative dissipativity properties between adjacent agents, distributed conditions are established to guarantee output synchronisation. Specifically, these conditions can be verified using only local information associated with neighbouring agents and coupling links. As an illustration, a heterogeneous network of Goodwin oscillators is considered, where the relative dissipativity properties between neighbouring oscillators are characterised and used to analyse synchronisation.
\end{abstract}


\section{Introduction}
Synchronisation in networks of dynamical systems has been extensively studied over the past decades due to its wide range of applications in engineering and natural systems, including multi-robot coordination \cite{olfati2007consensus}, sensor networks \cite{olfati2005consensus}, power systems \cite{dorfler2013synchronization}, and biological oscillators \cite{stan2007output}. In such settings, a collection of interacting agents aims to achieve a common behaviour through local information exchange governed by an underlying communication topology.

An important research direction studies synchronisation by exploiting \emph{passivity properties} of the agent dynamics. In these approaches, synchronisation can often be guaranteed by combining the passivity of individual agents with diffusive coupling protocols imposed by the network interconnection; see, e.g., \cite{chopra2006passivity,arcak2007passivity,chopra2012output,yu2013output}. A related line of work adopts an \emph{incremental viewpoint}, where dissipativity properties are formulated with respect to differences of trajectories rather than particular signals. In particular, the framework developed in \cite{scardovi2010synchronization} analyses the disagreement dynamics obtained by projecting the network behaviour onto the subspace orthogonal to the synchronisation manifold. This formulation leads to dissipativity conditions expressed in terms of relative outputs between agents and provides a convenient characterisation of synchronisation phenomena. Variants of this approach have subsequently been explored in several works; see, for example, \cite{hamadeh2011global,liu2011incremental} and the references therein.

However, many results based on such projected or disagreement-based formulations focus on \emph{networks of identical agents}, which simplifies the analysis of the projected dynamics. For example, the seminal work  \cite{scardovi2010synchronization} and several subsequent studies consider homogeneous networks. However, in many practical applications, agents are inherently heterogeneous due to differences in physical parameters, operating conditions, or modelling uncertainties. Extending synchronisation results to such settings is therefore important but technically challenging, since the symmetry properties that facilitate the analysis in homogeneous networks are no longer available. Extensions to heterogeneous networks have been investigated in more recent works, including \cite{franci2011input,alekajbaf2015input} and our previous work \cite{su2025output}, where synchronisation conditions are derived for networks with non-identical agent dynamics.

Nevertheless, these conditions, including those for homogenous and heterogeneous networks, are typically \emph{centralised}, in the sense that their verification requires global information about the network or collective properties of the agents.  In large-scale networked systems, such requirements may be restrictive due to scalability limitations and reduced agility in responding to network changes \cite{welikala2024decentralized}. 
From both a theoretical and a practical standpoint, it is therefore desirable to develop \emph{distributed synchronisation conditions} whose verification relies solely on local quantities associated with pairs of adjacent agents. Motivated by this objective, this paper studies the synchronisation problem for networks of heterogeneous nonlinear systems interconnected through nonlinear diffusive couplings. By exploiting relative dissipativity properties between \emph{neighbouring} agents, we derive \emph{distributed conditions} that guarantee synchronisation using only local information associated with adjacent agents and coupling links. To illustrate the proposed results, we consider a heterogeneous network of Goodwin oscillators as a case study. The relative dissipativity properties between neighbouring oscillators are first characterised, and the resulting distributed conditions are then used to analyse the network synchronisation behaviour.

The notation used throughout this letter is summarised as follows. Let $\mathbb{R}$ be the set of real numbers. Given a matrix $A$, let $A^\top$ denote its transpose, the notation $A\succ 0$ (respectively, $A\succeq  0$)  indicates that $A$ is positive definite (respectively, positive semidefinite). Denote ${\rm col}\left( {{a_1}, \ldots ,{a_m}} \right) := {\left[ {a_1, \ldots ,a_m} \right]^{\top}}$ as the column vector with scalars ${{a_1}, \ldots ,{a_m}}$. Let $\mathrm{diag}\{v_1,\dots, v_m\}$ denote the diagonal matrix with diagonal entries $v_1, \ldots, v_m$, and let $\mathrm{diag}(A)$ denote the operator that extracts the diagonal entries of a matrix and forms a diagonal matrix from them. Denote by $\Ltwo^m$ the space of signals $x:\left[ {0,\infty } \right) \to {\mathbb{R}^m}$ satisfying $\int_0^\infty  {{{\left| {x(t)} \right|}^2}dt < \infty}$, in the Lebesgue sense, with $|\cdot|$ being the Euclidean norm. Define $\Ltwoe^m: =\{ {x:\left[ {0,\infty } \right) \to {\mathbb{R}^m}\;|\;{P_T}x \in \Ltwo^m ,\forall T \ge 0} \}$, where $P_T$ is the truncation operator that satisfies $\left( {{P_T}x} \right)(t) = x(t)$ for $t \le T$ and $\left( {{P_T}x} \right)(t) = 0$ for $t>T$. For $x,y \in \Ltwoe^m $ and $T \ge 0$, ${\left\| x \right\|_T} := {\left( {\int_0^T {{{\left| {x(t)} \right|}^2}dt} } \right)^{1/2}}$ and ${\left\langle {x,y} \right\rangle _T} := \int_0^T {{x^{\top}}(t)y(t)dt}$. An operator $G:\Ltwoe^m \to \Ltwoe^m$ is said to be causal if ${P_T}G{P_T} = {P_T}G$ for all $T \ge 0$. Finally, For a finite set $\mathcal{N}$, $|\mathcal{N}|$ denotes the cardinality  of the set.

\section{Preliminaries and Problem Formulation}\label{sec: problem formulation}
\subsection{Graph Theory}
Let $\mathcal{G} = (\mathcal{N},\mathcal{E})$ be an undirected graph with node set
$\mathcal{N} = \{1,2,\ldots,n\},$
and edge set $\mathcal{E} = \{ e_1, e_2, \ldots, e_p \},$
where each edge $e_k$ connects nodes $(i,j)$. As the graph is undirected, for every $e_k=(i,j)$, there exists $k^*$ such than $e_{k^*}=(j,i)$. For each node $i$, let $\mathcal{E}_i$ denote the set of edges that include the node $i$, and let $\mathcal{N}_i$ denote the set of neighbours of node $i$. For an undirected graph $\mathcal{G}$, we assign an arbitrary orientation to each undirected edge. Let
$\mathcal{E}_i^{+}$ ($\mathcal{E}_i^{-}$)
denote the set of edges for which node $i$ is the positive (negative) endpoint under the chosen orientation. The heterogeneous network structure in Fig.~\ref{fig.network} can be represented by the incidence matrix $D = [d_{ik}] \in \mathbb{R}^{n \times p}$, where 
\[
d_{ik} = 
\begin{cases}
\;\;1, & \text{if edge } e_k \in \mathcal{E}_i^{+}, \\[4pt]
-1,   & \text{if edge } e_k \in \mathcal{E}_i^{-}, \\[4pt]
\;\;0, & \text{otherwise.}
\end{cases}
\]
This convention ensures that each column of $D$ corresponds to an edge and has exactly one $+1$ and one $-1$, representing its endpoints, while all other entries are zero.
Although the graph is undirected, assigning an arbitrary orientation to each edge is a standard device that allows the incidence matrix to be defined consistently. 

For a simple undirected graph $\mathcal{G}=(\mathcal{N},\mathcal{E})$ and any 
node $i\in\mathcal{N}$, we define the neighbour set
\[
\mathcal{N}_i := \{\, \ell\in\mathcal{N} : (i,\ell)\in\mathcal{E} \,\},
\qquad
r_i := |\mathcal{N}_i|.
\]
For an edge $e_k=(i,j)\in\mathcal{E}$, the set of \emph{common neighbours} of 
$i$ and $j$ is
$\mathcal{N}_{ij} := \mathcal{N}_i \cap \mathcal{N}_j,$
and $r_{ij} := |\mathcal{N}_{ij}|$.
We further define the \emph{exclusive neighbour set}
\[
\mathcal{N}_i^{\,j}
:= \{\, \ell \in \mathcal{N}_i : \ell \notin \mathcal{N}_{ij},\;\ell \neq j \,\}.
\]
Consequently, the total number of neighbours that belong to only one of 
either $i$ or $j$ (with $(i,j)\in\mathcal{E}$ so that $i$ and $j$ are adjacent, 
and excluding $i$ and $j$ themselves) is
\[
\bar r_{ij}
\;=\;
|\mathcal{N}_i^{\,j}| + |\mathcal{N}_j^{\,i}|
\;=\;
r_i + r_j - 2 r_{ij} - 2.
\]
Then, for each edge $e_k=(i,j)$, let us define $\tilde{r}_k=r_{ij}$, $\bar r_{k}=\bar r_{ij}$.
For a given graph $\mathcal{G}$, let us defined the following matrices:
\[
\begin{aligned}
\Phi_{\mathcal{G}} &= \operatorname{diag}\{\tilde r_1,\dots,\tilde r_p\}, \quad
\bar{\Phi}_{\mathcal{G}} = \frac{1}{2}\operatorname{diag}\{\bar r_1,\dots,\bar r_p\}.
\end{aligned}
\] 

The following proposition presents a positive definiteness property based on the incidence matrix for the subsequent synchronisation analysis.
\begin{proposition}\cite{su2026consensus}\label{prop: positive definite}
Given an undirected and connected graph $\mathcal{G} = (\mathcal{N}, \mathcal{E})$, let  $D\in\mathbb{R}^{n\times p}$  be  its incident matrix, and let $\Omega  =\mathrm{diag}\left\{ {{\mu _1}, \dots, {\mu _n}} \right\}$ and $\Sigma =\mathrm{diag}\{\sigma_1,\ldots,\sigma_p\}$.
It holds that  
\begin{align*}
    M := D^{\top}\Omega D+\Sigma \succ 0
\end{align*}
if for each $e_k=\left( {i,j} \right)\in \mathcal{E}$, 
\begin{align*}
{\sigma _k} + {\mu _i} + {\mu _j} - \left( {{r_i} - 1} \right)\left| {{\mu _i}} \right| - \left( {{r_j} - 1} \right)\left| {{\mu _j}} \right| > 0.
\end{align*}
\end{proposition}

\subsection{Problem Formulation}
At each node $i$, consider a dynamical system described by the causal operator $G_i:\Ltwoe \to\Ltwoe$ given by
\begin{equation}\label{eq: system model}
    {y_i} = {G_i}{u_i}, \,i\in\{1,2,\ldots,n\},
\end{equation}
where $u_i, y_i \in \Ltwoe$ denote respectively the input and output of the $i$-th system. The input $u_i$ to the $i$-th system is given by
\begin{equation}\label{eq: input} 
u_i = -\sum\limits_{j \in \mathcal{N}_i} {{{\vartheta}_{k}}\left( {y_i - y_j + w_{ij}} \right)}, \quad i \in \mathcal{N} 
\end{equation}
where the signal $w_{ij} \in \Ltwoe$  represents the disturbance present on the links connecting the $i$-th and $j$-th systems with $e_k=(i,j) \in \mathcal{E}$, and the operator ${\vartheta}_{k}:\Ltwoe \to \Ltwoe$.

Let $Y := {\rm col}\left( {{y_1}, \ldots ,{y_n}} \right)$ and define the stacked vectors $U$ and $W$ analogously. Using the incidence matrix representation, the distributed input law can be written compactly as
\begin{equation}\label{eq: iuput vector}
U = - D\Theta\left(D^{\top} Y+W \right), 
\end{equation}
where $\Theta: \Ltwoe^p \to \Ltwoe^p$ is the lower diagonal nonlinearity in Fig. \ref{fig.network}, hence we can write $V=\Theta\left(D^{\top}Y + W \right)$, where $V:={\rm col}\left( {{v_1}, \ldots ,{v_p}} \right)$. Each diagonal entry is given by  $v_k={\vartheta}_k((W+D^\top Y)_k)$.

In this work, the couplings and the heterogeneous agents are assumed to satisfy the following conditions.

\begin{assumption}\label{assum: coupling}
For each edge $e_k=(i,j)$, the operator ${\vartheta}_{k}:\Ltwoe \to \Ltwoe$ is a time-invariant memoryless sector-bounded nonlinearity characterised by the nonlinear function ${\vartheta}_{k}:\mathbb{R} \to \mathbb{R}$ \footnote{With a slight abuse of notation, $\vartheta_k$ denotes both the operator and its associated nonlinear function.}, 
satisfying $\vartheta_{k}(0)=0$ and the following properties:
\begin{enumerate}
\item for all $x\in\mathbb{R}$, ${{\vartheta _{k^*}}\left(-x\right)} = -  {{\vartheta _{k}}( x ) } $, which reflects the undirected structure of the graph; and
\item there exist constants $0<\underline{\alpha}_{k}\le\overline{\alpha}_{k} < \infty$ such that 
$\underline{\alpha}_{k}  \le \frac{\vartheta _{k}\left( x\right)}{x} \le \overline{\alpha}_{k},\,\forall x\neq 0.$
\end{enumerate}
\end{assumption}
\vskip2mm
Let $\overline{\Lambda}=\operatorname{diag}\{\overline{\alpha}_{1},\dots,\overline{\alpha}_{p}\}$ and $\underline{\Lambda}=\operatorname{diag}\{\underline{\alpha}_{1},\dots,\underline{\alpha}_{p}\}$.

\begin{assumption}\label{def: MI_OFP}
For each edge $e_k=(i,j)$, the operators $G_i,G_j:\Ltwoe \to \Ltwoe$ satisfy
\begin{align} \label{eq: MI_OFP}
&\left\langle u - v,\; G_i u - G_j v \right\rangle_T \nonumber\\
&\quad \ge \nu_{k}\!\left( \|u\|_T^2 + \|v\|_T^2 \right)
+ \gamma_k\|G_i u - G_j v\|_T^2 + \beta_k,
\end{align}
for some constants $\nu_{k}\le 0$, $\gamma_{k}\in\mathbb{R}$, and $\beta_{k}\in\mathbb{R}$, and for all $u,v \in \Ltwoe$ and $T \ge 0$.
\end{assumption}

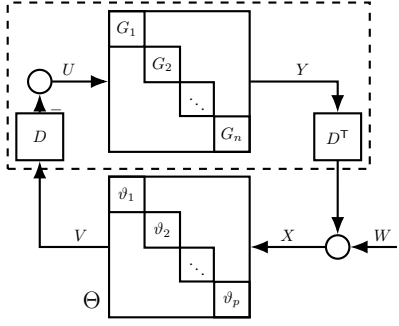
\begin{figure}[ht]
\centering
\begin{tikzpicture}[scale=0.62,every node/.style={transform shape},
  >=Latex,
  block/.style={draw, thick, minimum width=1.0cm, minimum height=1.0cm},
  bigblock/.style={draw, thick, minimum width=3.0cm, minimum height=3cm},
  sum/.style={draw, circle, thick, minimum size=5mm, inner sep=0pt},
  lab/.style={font=\normalsize},
  tinybox/.style={draw, thick, minimum width=0.95cm, minimum height=0.75cm, inner sep=0pt},
  line/.style={thick}
]

\def\gapLeft{1.2cm}   
\def\gapRight{1.6cm}  
\def\vgap{0.5cm}      
\def\dvert{4mm}       

\node[bigblock] (Hblk) {};

\node[bigblock, below=\vgap of Hblk] (Tblk) {};

\node[sum, left=\gapLeft of Hblk]  (sumU) {};
\node[sum, right=\gapRight of Tblk] (sumL) {};

\node[block, below=\dvert of sumU] (D) {$D$};

\node[block] (DT) at ($(sumL |- D)$) {$D^{\mathsf T}$};

\node at ($ (Tblk.south west)!0.25!(Tblk.west) + (-10pt,0) $) {\Large{$\Theta$}};


\path let
  \p1 = (Hblk.north west),
  \p2 = (Hblk.south east),
  \n1 = {(\x2-\x1)/4},   
  \n2 = {(\y1-\y2)/4}    
in
  node[tinybox, anchor=north west, minimum width=\n1, minimum height=\n2] (H1)
       at (Hblk.north west) {$G_{1}$}

  node[tinybox, anchor=north west, minimum width=\n1, minimum height=\n2] (H2)
       at ($ (Hblk.north west) + (\n1,-\n2) $) {$G_{2}$}

  node[tinybox, anchor=north west, minimum width=\n1-1, minimum height=\n2-1] (Hdd)
       at ($ (Hblk.north west) + (2*\n1,-2*\n2) $) {$\ddots$}

  node[tinybox, anchor=south east, minimum width=\n1, minimum height=\n2] (Hn)
       at (Hblk.south east) {$G_{n}$};


\path let
  \p3 = (Tblk.north west),
  \p4 = (Tblk.south east),
  \n3 = {(\x4-\x3)/4},   
  \n4 = {(\y3-\y4)/4}    
in
  node[tinybox, anchor=north west, minimum width=\n3, minimum height=\n4] (L1)
       at (Tblk.north west) {$\vartheta_{1}$}

  node[tinybox, anchor=north west, minimum width=\n3, minimum height=\n4] (L2)
       at ($ (Tblk.north west) + (\n3,-\n4) $) {$\vartheta_{2}$}

  node[tinybox, anchor=north west, minimum width=\n3-1, minimum height=\n4-1] (Ldd)
       at ($ (Tblk.north west) + (2*\n3,-2*\n4) $) {$\ddots$}

  node[tinybox, anchor=south east, minimum width=\n3, minimum height=\n4] (Ln)
       at (Tblk.south east) {$\vartheta_{p}$};

\draw[line, ->] (D.north) -- (sumU.south);

\node[lab] at ($(sumU.south)+(0.35,-0.35)$) {$-$};

\draw[line, ->] (sumU.east) -- ($(Hblk.west)+(0,0)$) node[pos=0.3, above, lab] {$U$};

\draw[line, ->] (Hblk.east) -| node[pos=0.3, above, lab]{$Y$} (DT.north);

\draw[line, <-] (sumL.east) -- ++(1,0) node[pos=0.7, above, lab] {$W$};

\draw[line, ->]   (sumL.west) -- node[pos=0.5, above, lab] {$X$} (Tblk.east);

\draw[line,->] (Tblk.west) 
  -| 
  node[pos=0.2, above, lab] {$V$}
  (D.south);
\draw[line, ->] (DT.south) --node[pos=0.4, left, lab] {} (sumL.north);

\node[draw, thick, dashed, fit={(Hblk) (D) (DT)}, inner sep=3pt] (Hgroup) {};


\end{tikzpicture}
\vskip-0.6cm
\caption{Block diagram of the network described by \eqref{eq: system model} \& \eqref{eq: input}.}
\label{fig.network}
\vskip-0.8cm
\end{figure} 



The operators $G_i$ and $G_j$ are said to be $(\nu_{ij},\gamma_{ij})$-\emph{generalised mutually incremental output-feedback passive} if \eqref{eq: MI_OFP} holds \cite{su2025output}. 
This condition characterises a relaxed relative passivity relation between pairs of agents in the network and accommodates heterogeneous nonlinear dynamics. 
In particular, when $\nu_k=0$, the condition reduces to a \emph{mutually relaxed cocoercivity} relation between $G_i$ and $G_j$ (cf. \cite{alekajbaf2015input}). 
Moreover, when $\nu_k=0$ and $G_i = G_j$, it further reduces to \emph{incremental output-feedback passivity} property (cf. \cite{van2000l2}). 
As shown in \cite{su2025output}, this condition holds for several classes of nonlinear systems. Additionally, 
when the operators $G_i$ and $G_j$ correspond to linear time-invariant (LTI) systems, condition \eqref{eq: MI_OFP} admits an equivalent characterisation in terms of linear matrix inequalities (LMIs); see \cite{su2025output} for details.

\section{Main Results}\label{sec: Main Results}


This section establishes the distributed conditions that guarantee synchronisation of the network shown in Fig.~\ref{fig.network}. To this end, we first present a technical lemma that reformulates the property in Assumption~\ref{def: MI_OFP} as a dissipativity condition with respect to the signal pair $(V, D^\top Y)$, capturing the open-loop operator mapping $V$ to $D^\top Y$ in the dashed box in Fig.~\ref{fig.network}.

\begin{lemma}\label{lem: U,Y}
Consider the network described by \eqref{eq: system model} and \eqref{eq: input}. If the Assumption~\ref{def: MI_OFP} holds, then the following inequality is satisfied:
\begin{align*}
&{\left\langle {-V,\left( {2I + \Phi_{\mathcal{G}}} \right){D^ \top }Y} \right\rangle _T}\\
\ge& {\left\langle {{D^ \top }Y,\left({\Gamma  - \bar \Phi_{\mathcal{G}}}\right){D^ \top }Y} \right\rangle _T}  + {\left\langle {V,\left( {{D^ \top }\Xi D - \bar \Phi_{\mathcal{G}}} \right)V} \right\rangle _T} + \bar \beta,   
\end{align*}
where $\bar \beta  =\sum\nolimits_{k = 1}^p {{\beta _k}} $, 
$\Gamma  =\mathrm{diag}\left\{ {{\gamma_1}, \dots, {\gamma_p}} \right\}$, and $\Xi =\mathrm{diag}\left\{ {{\tilde\nu _1}, \dots, {\tilde\nu _n}} \right\}$ with $\tilde\nu_i = \sum\nolimits_{e_k \in {\mathcal{E}_i}} {{\nu _{k}}}$.
\end{lemma}

\begin{proof}
From Assumption~\ref{def: MI_OFP}, it follows that
\begin{align}\label{eq: thm2_DU,DY}
&\quad{\left\langle {{D^\top}U,{D^\top}Y} \right\rangle _T} 
=\sum_{(i,j) \in {\mathcal{E}}} {{{\left\langle {{y_i} - {y_j},{u_i} - {u_j}} \right\rangle }_T}} \nonumber\\
&\ge\sum_{e_k=(i,j) \in \mathcal{E}} {{\gamma _k}\left\| {{y_i} - {y_j}} \right\|_T^2}  + \sum\limits_{e_k=(i,j) \in \mathcal{E}} {{\beta _k}} \nonumber\\
&\quad+\sum\limits_{e_k=(i,j) \in \mathcal{E}} {{\nu _k}\left( {\left\| {{u_i}} \right\|_T^2 + \left\| {{u_j}} \right\|_T^2} \right)} \nonumber\\
&= {\left\langle {{D^ \top }Y,\Gamma {D^ \top }Y} \right\rangle _T} + {\left\langle {U,\Xi U} \right\rangle _T} + \bar \beta\nonumber\\
&= {\left\langle {{D^ \top }Y,\Gamma {D^ \top }Y} \right\rangle _T} + {\left\langle {V,D^\top \Xi DV} \right\rangle _T} + \bar \beta.
\end{align}
Let $y_{ij}:=y_i-y_j$. 
Since $(DD^\top Y)_i=(LY)_i = \sum\nolimits_{l \in {\mathcal{N}_i}} {{y_{il}}}$, the following identity holds for any $e_k=(i,j) \in \mathcal{E}$:
\begin{align*}
&(D^\top DD^\top Y)_k =
\sum\limits_{l \in {\mathcal{N}_i}} {y_{il}}  - \sum\limits_{l \in {\mathcal{N}_j}} {y_{jl}} \\
& = {y_{ij}} - {y_{ji}} + \sum\limits_{l \in {{\mathcal{N}}_{ij}}} {{y_{il}}}  - \sum\limits_{l \in {{\mathcal{N}}_{ij}}} {{y_{jl}}}+\sum\limits_{l\in\mathcal{N}_i^{ j}} {y_{il}}  - \sum\limits_{l\in\mathcal{N}_j^{ i}} {y_{jl}}\\
&= \left( {2 + { r_{ij}}} \right) y_{ij}+\sum\limits_{l\in\mathcal{N}_i^{ j}} {y_{il}}  - \sum\limits_{l\in\mathcal{N}_j^{ i}} {y_{jl}}.
\end{align*}
Consequently,
\begin{align}\label{eq: thm2_U,DD-IY_1}
&- {\left\langle {U,\left( {D{D^ \top } - I} \right)Y} \right\rangle _T}
= {\left\langle {V,{D^ \top }\left( {D{D^ \top } - I} \right)Y} \right\rangle _T}\nonumber\\
& =\sum\limits_{e_k=(i,j) \in {\mathcal{E}}} {{{\left\langle {v_k,- y_{ij} + \sum\limits_{l \in {\mathcal{N}_i}} {y_{il}}  - \sum\limits_{l \in {\mathcal{N}_j}} {y_{jl}} } \right\rangle }_T}}\nonumber\\
&=\sum\limits_{e_k=(i,j) \in \mathcal{E}} {{{\left\langle {v_k,\sum\limits_{l\in\mathcal{N}_i^{ j}} {y_{il}}  - \sum\limits_{l\in\mathcal{N}_j^{ i}} {y_{jl}} } \right\rangle }_T}}\nonumber\\
&\quad +\sum\limits_{e_k=(i,j) \in \mathcal{E}} {\left( {1 + {r_{ij}}} \right){{\left\langle {v_k,y_{ij}} \right\rangle }_T}}
\end{align}
For each $e_k=(i,j) \in \mathcal{E}$, 
it holds from $\bar r_{ij}
\;=\;
|\mathcal{N}_i^{\,j}| + |\mathcal{N}_j^{\,i}|$ and Young's inequality that
\begin{align*}
&{{\left\langle {v_k,\sum\limits_{l\in\mathcal{N}_i^{ j}} {y_{il}}  - \sum\limits_{l\in\mathcal{N}_j^{ i}} {y_{jl}} } \right\rangle }_T}\nonumber\\ 
&\ge  -\frac{\bar r_{ij}}{2}\left\| {v_k} \right\|_T^2-\frac{1}{2}\sum\limits_{l\in\mathcal{N}_i^{ j}} \left\|{y_{il}}\right\|_T^2  - \frac{1}{2}\sum\limits_{l\in\mathcal{N}_j^{ i}} \left\|{y_{jl}}\right\|_T^2.
\end{align*}
Furthermore, we have
\begin{align*}
&\sum\limits_{(i,j) \in \mathcal{E}} {\left( {\sum\limits_{l\in\mathcal{N}_i^{ j}} \left\|{{y_{il}}}\right\|_T^2  + \sum\limits_{l\in\mathcal{N}_j^{ i}} \left\|{{y_{jl}}}\right\|_T^2 } \right)} = \sum\limits_{(i,j) \in \mathcal{E}} {{{\bar r}_{ij}}\left\|{y_{ij}}\right\|_T^2}.
\end{align*}
Using the above two equations to lower bound  $\sum_{e_k=(i,j) \in \mathcal{E}} {{{\left\langle {v_k,\sum_{l\in\mathcal{N}_i^{ j}} {y_{il}}  - \sum_{l\in\mathcal{N}_j^{ i}} {y_{jl}} } \right\rangle }_T}}$, we obtain  from \eqref{eq: thm2_U,DD-IY_1}  that
\begin{align}\label{eq: thm2_U,DD-IY}
&- {\left\langle {U,\left( {D{D^ \top } - I} \right)Y} \right\rangle _T}
\nonumber\\
\ge & - \sum\limits_{e_k=(i,j) \in \mathcal{E}} {\left(\frac{\bar r_{ij}}{2}\left\| {v_k} \right\|_T^2+\frac{1}{2}\sum\limits_{l\in\mathcal{N}_i^{ j}} \left\|{y_{il}}\right\|_T^2  + \frac{1}{2}\sum\limits_{l\in\mathcal{N}_j^{ i}} \left\|{y_{jl}}\right\|_T^2\right)} \nonumber\\
& +\sum\limits_{e_k=(i,j) \in \mathcal{E}} {\left( {1 + {r_{ij}}} \right){{\left\langle {v_k,y_{ij}} \right\rangle }_T}} \nonumber\\
=& \sum\limits_{{e_k}=(i,j) \in \mathcal{E}} {\left( {\left( {1 + {{ r}_{ij}}} \right){{\left\langle {{v_k},{ y_{ij}}} \right\rangle }_T} - \frac{{{{\bar r}_{ij}}}}{2}\left\| {{v_k}} \right\|_T^2 - \frac{{{{\bar r}_{ij}}}}{2}\left\| {{\tilde y_k}} \right\|_T^2} \right)} \nonumber\\
=& {\left\langle {V,\left( {I + \Phi_{\mathcal{G}}} \right){D^ \top }Y} \right\rangle _T} - {\left\langle {V,\bar \Phi_{\mathcal{G}} V} \right\rangle _T} - {\left\langle {{D^ \top }Y,\bar \Phi_{\mathcal{G}}{D^ \top }Y} \right\rangle _T}.
\end{align}
Summing up \eqref{eq: thm2_DU,DY} and \eqref{eq: thm2_U,DD-IY} yields that
\begin{align*}
&{\left\langle {-V,D^{\top} Y} \right\rangle _T}={\left\langle {U,Y} \right\rangle _T}\\
= &{\left\langle {{D^ \top }U,{D^ \top }Y} \right\rangle _T} - {\left\langle {U,\left( {D{D^ \top } - I} \right)Y} \right\rangle _T}\nonumber\\
\ge &{\left\langle {{D^ \top }Y,\left({\Gamma  - \bar \Phi_{\mathcal{G}}}\right){D^ \top }Y} \right\rangle _T}+ {\left\langle {V,\left( {I + \Phi_{\mathcal{G}}} \right){D^ \top }Y} \right\rangle _T} \nonumber\\
&\quad + {\left\langle {V,\left( {{D^ \top }\Xi D - \bar \Phi_{\mathcal{G}}} \right)V} \right\rangle _T} + \bar \beta,
\end{align*}
which leads to
\begin{align*}
&{\left\langle {-V,\left( {2I + \Phi_{\mathcal{G}}} \right){D^ \top }Y} \right\rangle _T}\ge\\
& {\left\langle {{D^ \top }Y,\left({\Gamma  - \bar \Phi_{\mathcal{G}}}\right){D^ \top }Y} \right\rangle _T}  + {\left\langle {V,\left( {{D^ \top }\Xi D - \bar \Phi_{\mathcal{G}}} \right)V} \right\rangle _T} + \bar \beta.
\end{align*}
\end{proof}

\begin{remark}
Lemma \ref{lem: U,Y} converts the pairwise dissipativity conditions satisfied by adjacent agents into a network-level dissipativity inequality expressed in terms of the signal pair $(V, D^\top Y)$ by exploiting the graph structure. This inequality characterises the dissipativity property of the  operator mapping $V$ to $D^\top Y$ as shown in the dashed box in Figure \ref{fig.network}. In particular, $D^\top Y$ represents the relative outputs between neighbouring agents, while $V$ collects the outputs of the nonlinear coupling operators associated with each edge of the network. Hence, the resulting inequality is formulated entirely in terms of edge variables. In combination with the dissipativity of the other subsystem in Fig.~\ref{fig.network}, namely $\mathrm{diag}\{\vartheta_1,\ldots,\vartheta_p\}$, Theorem~\ref{thm: consensus condition} applies the dissipativity theorem to establish stability of the feedback interconnection and thus ensure synchronisation.
\end{remark}

\begin{theorem}\label{thm: consensus condition}
Consider the network described by \eqref{eq: system model} and \eqref{eq: input}. Suppose the Assumptions~\ref{assum: coupling} and~\ref{def: MI_OFP} hold and  the following condition holds for all edges $e_k=(i, j) \in \mathcal{E}$:
\[
   \frac{2+\tilde{r}_{k}}{\overline{\alpha}_{k}} -\frac{{  (1 + \underline{\alpha} _k^2){{\bar r}_k}}}{{2\underline{\alpha} _k^2}}+\frac{{{\gamma _{k}}}}{{{\underline{\alpha}_{k} ^2}}}-r_i |\tilde\nu_i| - r_j|\tilde\nu_j| > 0,
\]
where $\tilde\nu_i = \sum\nolimits_{e_k \in {\mathcal{E}_i}} {\nu _k}$. Then, there exist a finite gain $\rho > 0$ and a constant $\epsilon \ge 0$ such that 
\begin{align}\label{eq:aim}
{\left\| {D^{\top} Y} \right\|_T}\le \rho{\left\| { W} \right\|_T}+\epsilon,\,\forall {W } \in \Ltwoe^p,\,\forall T \ge 0.
\end{align}
\end{theorem}

\begin{proof}
By Assumption~\ref{assum: coupling}, for each $k=1,\dots,p$,
${\vartheta _k}\left( {x(t)} \right)^2 \le {\overline{\alpha}_k}x(t){\vartheta _k}\left( {x(t)} \right)$.
Integrating over $[0,T]$ yields
$$\overline{\alpha}_k^{-1}\left\| {{{\vartheta} _k}\left( x \right)} \right\|_T^2 \le {\left\langle {x,{{\vartheta} _k}\left( x \right)} \right\rangle _T},\,\forall x\in\Ltwoe$$
In terms of the diagonal operator $\Theta$, it gives
$${\left\langle {X,\Theta (X)} \right\rangle _T} \ge {\left\langle {\Theta (X),\overline{\Lambda}^{-1} \Theta (X)} \right\rangle _T},\,\forall X\in\Ltwoe^p$$
where $\overline{\Lambda}$ is defined after Assumption~\ref{assum: coupling}. Let $X:=D^{\top}Y + W$ and recall that $V=\Theta(X)$. Then, due to the diagonal structure of $I$ and $\Phi_{\mathcal{G}}$, one has
\begin{align}\label{eq: thm2_V,Y}
{\left\langle {V,{\left( {2I + \Phi_{\mathcal{G}} } \right)}X} \right\rangle _T} \ge {\left\langle {V,{\left( {2I + \Phi_{\mathcal{G}} } \right)}\overline{\Lambda}^{-1} V} \right\rangle _T}.
\end{align}
It follows from \eqref{eq: thm2_V,Y} and Lemma~\ref{lem: U,Y} that
\begin{align}\label{eq: thm2_V,W}
&{\left\langle {V,\left( {2I + \Phi_{\mathcal{G}} } \right)W} \right\rangle _T}\nonumber\\
=& {\left\langle {-V, \left( {2I +\Phi_{\mathcal{G}} } \right){D^{\top}}Y} \right\rangle _T} + {\left\langle {V,\left( {2I + \Phi_{\mathcal{G}} } \right)X} \right\rangle _T}\nonumber\\
\ge& {\left\langle {V,\Psi V} \right\rangle _T}+ {\left\langle {{D^ \top }Y,\left( {\Gamma  - \bar \Phi_{\mathcal{G}} } \right){D^ \top }Y} \right\rangle _T} + \bar \beta,
\end{align}
where $\Psi:=D^\top \Xi D-\bar\Phi_G+(2I+\bar\Phi_\mathcal{G})\overline{\Lambda}^{-1}$, $\Xi$, $\Gamma$ and $\bar \beta$ are defined in Lemma~\ref{lem: U,Y}.
Next, for each $e_k=(i,j)\in\mathcal{E}$, there exists measurable scalar function $\eta_k(t) $ such that $\underline{\alpha}_k \le \eta_k(t) \le \overline{\alpha}_k$ and $\eta_k(t) \cdot \left( {{y_i}(t) - {y_j}(t)+w_{ij}(t)} \right)={{\vartheta_k}\left( {{y_i}(t) - {y_j}(t)+w_{ij}(t)} \right)} $ for all $t\ge 0$. Let $H\left(t\right):= \mathrm{diag}\{ {\eta_1}(t), \ldots ,{\eta_p}(t)\}$. Then $V(t)=H(t)X(t)$.
Define $M(t):=H(t)\Psi H(t),\,t\ge 0$.
It follows that
\begin{align*}
&{\left\langle {V,\Psi V} \right\rangle _T}
={\left\langle {{D^ \top }Y + W, M\left( {D^ \top }Y + W\right)} \right\rangle _T}\nonumber\\
&={\left\langle {{D^ \top }Y,M {D^ \top }Y} \right\rangle _T} + 2{\left\langle {{D^ \top }Y,M W} \right\rangle _T} + {\left\langle {W,M W} \right\rangle _T}.
\end{align*}
Substituting the above equation into \eqref{eq: thm2_V,W}, and defining $N(t):= M(t)+\Gamma-\bar\Phi_\mathcal{G}$,
we obtain
\begin{align}\label{eq: thm2_V,W_1}
\quad{\left\langle {V,\left( {2I + \Phi_\mathcal{G} } \right)W} \right\rangle _T}&\ge {\left\langle{{D^{\top}}Y,N{D^{\top}}Y} \right\rangle _T}+ {\left\langle {W,M W} \right\rangle _T}\nonumber\\
&\quad + 2{\left\langle {{D^{\top}}Y,M W} \right\rangle _T}+ \bar \beta
\end{align}
Since when $\gamma_k>0$, inequality \eqref{eq: MI_OFP} still holds with $\gamma_k$ replaced by $0$, we assume without loss of generality that $\gamma_k \le 0$ for $k=1,\dots,p$. As a result, $\Gamma  - \bar \Phi_{\mathcal{G}}\preceq 0$. Now define $Q:=\Psi+\underline{\Lambda}^{-1}(\Gamma-\bar\Phi_\mathcal{G})\underline{\Lambda}^{-1}$ with $\underline{\Lambda}$ is defined after Assumption~\ref{assum: coupling}. Since $\eta_k(t)\ge \underline\alpha_k$ for all $k$ and all $t\ge 0$, we have $H(t)^{-1}-\underline{\Lambda}^{-1}\preceq 0$. Therefore, if $Q\succ0$, then $\Psi+H(t)^{-1}(\Gamma-\bar\Phi_\mathcal{G})H(t)^{-1}\succ 0$,
which implies $N\left( t \right)\succ 0$ for all $t\ge0$. By hypothesis, $\frac{2+\tilde{r}_{k}}{\overline{\alpha}_{k}} -\frac{{  (1 + \underline{\alpha} _k^2){{\bar r}_k}}}{{2\underline{\alpha} _k^2}}+\frac{{{\gamma _{k}}}}{{{\underline{\alpha}_{k} ^2}}}-r_i |\tilde\nu_i| - r_j|\tilde\nu_j| > 0$ for all $e_k=(i, j) \in \mathcal{E}$. Combining this with Proposition \ref{prop: positive definite} as well as the fact that $\nu_k\le0$, we conclude that $Q\succ 0$, and hence $ N(t)\succ0$ for all $t\ge 0$. Since also $\Gamma-\bar\Phi_\mathcal{G}\preceq 0$,
it follows that $ M(t)\succ0$ for all $t\ge 0$.
Define $\underline\mu:=\inf_{t\ge 0}\lambda_{\min}\bigl(N(t)\bigr)>0$ and $\overline\mu:=\sup_{t\ge 0}\lambda_{\max}\bigl(M(t)\bigr)$, where $\lambda_{\min}(\cdot)$ and $\lambda_{\max}(\cdot)$ denote the smallest and largest eigenvalues, respectively. By the Cauchy–Schwartz inequality, it follows from \eqref{eq: thm2_V,W_1} that 
\begin{align*}
{\left\langle {V,\left( {2I + \Phi_\mathcal{G} } \right)W} \right\rangle _T}\ge \underline\mu \left\| {{D^{\top}}Y} \right\|_T^2 - 2\overline\mu{\left\| {{D^{\top}}Y} \right\|_T}{\left\| { W} \right\|_T}+ \bar \beta,
\end{align*}
which leads to
\begin{align}\label{eq: thm_DY}
\underline\mu \left\| {{D^{\top}}Y} \right\|_T^2 \le {\left\langle {V,{\left( {2I +\Phi_\mathcal{G} } \right)}W} \right\rangle _T}
+ 2\overline\mu{\left\| {{D^{\top}}Y} \right\|_T}{\left\| {W} \right\|_T}- \bar \beta.
\end{align}
Let $\overline\phi:= \mathop {\max }\nolimits_{e_k \in \mathcal{E}} {(2+\tilde r_k)}$, $\overline{\alpha} := \mathop {\max }\nolimits_{e_k \in \mathcal{E}}{\overline{\alpha}_k}$. Using Young’s inequality, and by the bound on $\Theta$ induced by Assumption~\ref{assum: coupling},
\begin{align*}
\left\| V \right\|_T^2 \le {\overline{\alpha} ^2}\left\| {{D^{\top}}Y + W} \right\|_T^2 
\le 2{\overline{\alpha} ^2}\left\| {{D^{\top}}Y} \right\|_T^2 + 2{\overline{\alpha}^2}\left\| {W} \right\|_T^2.
\end{align*}
Using these estimates in \eqref{eq: thm_DY}, one has
\begin{align*}
&\quad\underline\mu \left\| {{D^{\top}}Y} \right\|_T^2 \\
&\le \frac{{\underline \mu }}{{8{\overline{\alpha} ^2}}}\left\| V \right\|_T^2 + \frac{{2{\overline{\alpha} ^2}}}{{\underline\mu }}\left\| {{\left( {2I + \Phi_\mathcal{G} } \right)}W} \right\|_T^2\\
&\quad+ \frac{{\underline\mu }}{4}\left\| {{D^{\top}}Y} \right\|_T^2 + \frac{{4\overline\mu^2}}{{\underline\mu }}\left\| {W} \right\|_T^2 - \bar \beta \nonumber\\
& \le \frac{{\underline \mu }}{2}\left\| {{D^{\top}}Y} \right\|_T^2 + \left( {\frac{{\underline \mu }}{4} + \frac{{2{\overline{\alpha} ^2}\overline\phi^2 + 4{\bar\mu ^2}}}{{\underline \mu }}} \right)\left\| {W} \right\|_T^2 - \bar \beta, 
\end{align*}
which leads to
\begin{align}\label{eq: thm2_norm Y}
\left\| {{D^{\top}}Y} \right\|_T^2 \le \left( {\frac{1}{2} + \frac{{4{\overline{\alpha} ^2}\overline\phi^2 + 8{\overline\mu ^2}}}{{{{\underline \mu }^2}}}} \right)\left\| {W} \right\|_T^2 - \frac{{2\bar \beta }}{{\underline \mu }}.
\end{align}
Finally, it follows from \eqref{eq: thm2_norm Y} and ${a^2} \pm {b^2} \le  {\left( {\left| a \right| + \left| b \right|} \right)^2}$ that
$${\left\| {{D^{\top}}Y} \right\|_T} \le \rho {\left\| {W} \right\|_T} + \epsilon,\,\forall {W } \in \Ltwoe^p,\,\forall T \ge 0,$$
where $\rho =\sqrt {\frac{1}{2} + \frac{{4{\overline{\alpha} ^2}\overline\phi^2 + 8{\overline\mu ^2}}}{{{{\underline \mu }^2}}}}>0$ and $\epsilon=\sqrt {\frac{{2\left| {\bar \beta } \right|}}{{\underline \mu }}}\ge0$.
\end{proof}

\begin{remark}
The inequality \eqref{eq:aim} characterises an input–output property of the network with respect to the external input $W$. Recalling the definition of the incidence matrix, the term $D^\top Y$ represents the output disagreement between adjacent agents. Hence, ${\left\| {{D^{\top}}Y} \right\|_T}$ measures the mismatch between the output differences of the adjacent agents over the interval $[0,T]$, and \eqref{eq:aim} describes the property that bounded external input $W$ results in bounded  ${\left\| {{D^{\top}}Y} \right\|_T}$ in any interval $[0,T]$. 
In particular, when $W=0$, the inequality reduces to ${\left\| {{D^{\top}}Y} \right\|_T}\le \epsilon$, which indicates that the outputs of the agents approach synchronisation.
\end{remark}


\section{A Case Study: Goodwin Oscillators}\label{sec:  Case Studies}
\begin{figure}[!ht]
\centering
\includegraphics[width=7.5cm]{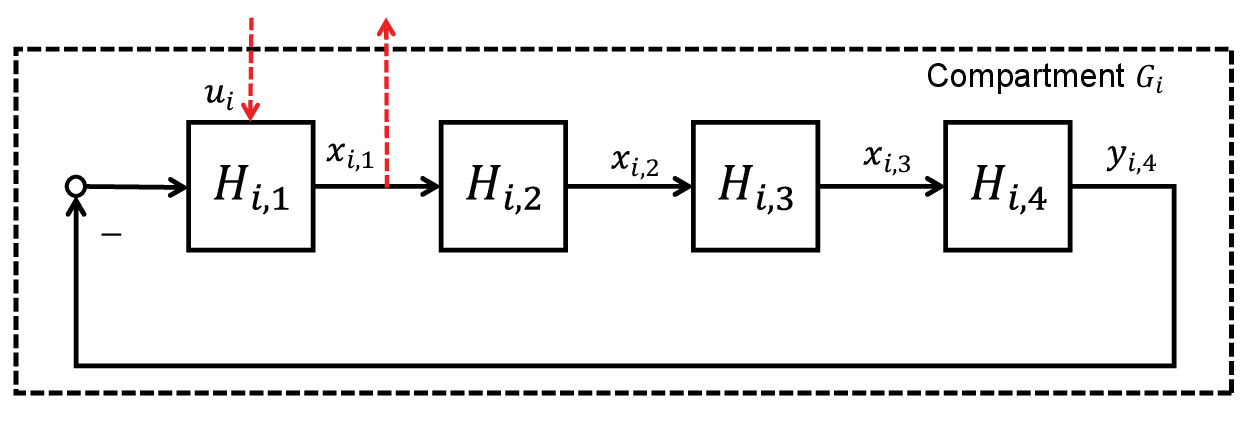}
\caption{Goodwin Oscillator $G_i$.}
\label{fig.Goodwin Oscillators}
\end{figure}
\vskip-3mm
The Goodwin oscillator is a classical model for describing oscillatory behaviour in genetic regulatory networks arising from negative feedback in gene expression. It captures the cyclic interaction between transcription, translation, and metabolite production, where the end product represses its own synthesis through nonlinear inhibition (see, e.g., \cite{stan2007output} and references therein).  Previous studies often consider networks of identical Goodwin oscillators, leading to homogeneous network structures (see, e.g., \cite{stan2007output,scardovi2010synchronization}). In contrast, real biological systems typically exhibit variability in regulatory interactions. Motivated by this, we consider a heterogeneous network of Goodwin oscillators in which the input gains may differ across agents, providing a more realistic representation of biological variability.  Specifically, each oscillator is modelled as a compartment consisting of four cyclically interconnected subsystems:
\begin{align}\label{eq: Goodwin_oscillator}
{G_i}:
\begin{cases}
{{{\dot x}_{i,1}} =  - {a_1}{x_{i,1}} - {y_{i,4}} + {b_{i,1}}{u_i}}\\
{{{\dot x}_{i,2}} =  - {a_2}{x_{i,2}} + {b_2}{x_{i,1}}}\\
{{{\dot x}_{i,3}} =  - {a_3}{x_{i,3}} + {b_3}{x_{i,2}}}\\
{{y_{i,4}} =  - \frac{1}{{x_{i,3}^p + 1}}}
\end{cases} 
\end{align}
for $i \in \{1,2,3,4,5\}$, where $a_1,a_2,a_3,b_2,b_3>0$, the coefficients $b_{i,1}>0$ are heterogeneous input gains, and $p$ is the Hill coefficient characterising the cooperativity of end-product repression. Let $\delta := \tfrac{p(p-1)}{(\sqrt[p-1]{((p-1)/(p+1))^p}+1)^2 (p+1)}$. The compartments are interconnected  through diffusive coupling of the first species via the input $u_i$, as illustrated in Fig.~\ref{fig.Goodwin Oscillators}, where the red dashed lines denote the coupling links. 
We begin by verifying Assumption~\ref{def: MI_OFP}.
\begin{proposition}\label{prop: case study 1}
Given two Goodwin oscillators  $G_i$ and $G_j$ as defined in \eqref{eq: Goodwin_oscillator},
Assumption \ref{def: MI_OFP} holds with ${\nu_k} =  - \frac{{\tilde b_{i,j}^2}}{{2\theta }}$, ${\gamma _k} = { {a_1} - {\theta } - \frac{{{\theta _1}}}{2}-\frac{{{\theta _2}}}{2}}$ where $\tilde b_{i,j}= \max\{\left| {{b_{i,1}} - 1} \right|,\left| {{b_{j,1}} - 1} \right|\}$, $\theta>0$ is arbitrary, and ${\theta _1} = \frac{{{\delta ^2}{\theta _3}}}{{2{a_3}{\theta _3} - b_3^2}}>0$, ${\theta _2} = \frac{{b_2^2}}{{2{a_2} - {\theta _3}}}>0$ with some $\theta_3>0$.
\end{proposition}

\begin{proof}
Since $f\left( x \right) =  - \frac{1}{{{x^p} + 1}}$ satisfies $\frac{d}{{dx}}f(x) \le \delta $ \cite[Example A]{scardovi2010synchronization},
then we have
$
\frac{1}{2}\frac{d}{{dt}}{( {{x_{i,1}} - {x_{j,1}}} )^2} 
= -a_1 \left( {{x_{i,1}} - {x_{j,1}}} \right)^2-\left( {{x_{i,1}} - {x_{j,1}}} \right)\left( {{y_{i,4}} - {y_{j,4}}} \right)+ \left( {{x_{i,1}} - {x_{j,1}}} \right)\left( {u_i- u_j} \right)+\left( {{x_{i,1}} - {x_{j,1}}} \right)( {\tilde b_{i,1}u_i- \tilde b_{j,1}u_j} ),
$
where $\tilde b_{i,1}=b_{i,1}-1$ and $\tilde b_{j,1}=b_{j,1}-1$. 
Applying the Mean Value Theorem to $f(x)$, we obtain $\left| {{y_{i,4}} - {y_{j,4}}} \right| = \left| {f({x_{i,3}}) - f({x_{j,3}})} \right| \le \delta \left| {{x_{i,3}} - {x_{j,3}}} \right|$. This, together with Young's inequality, gives
\begin{align*}
\frac{1}{2}\frac{d}{{dt}}{\left( {{x_{i,1}} - {x_{j,1}}} \right)^2} 
\le ( { - {a_1} + {\theta } + \frac{{{\theta _1}}}{2}} )\left( {{x_{i,1}} - {x_{j,1}}} \right)^2+\qquad\nonumber \\
\frac{{{\delta ^2}}}{{2{\theta _1}}}{\left( {{x_{i,3}} - {x_{j,3}}} \right)^2}
+\left( {{x_{i,1}} - {x_{j,1}}} \right)\left( {u_i- u_j} \right)+\frac{{\tilde b_{i,j}^2}}{{2\theta }}\left({u_i^2+u_j^2}\right)
\end{align*}
for any $\theta >0$ and $\theta_1 >0$. 
Similarly, one has
\begin{align*}
&\frac{1}{2}\frac{d}{{dt}}{\left( {{x_{i,2}} - {x_{j,2}}} \right)^2} \nonumber\\
\le & ( {-a_2+\frac{{b_2^2}}{{2{\theta _2}}}} )\left( {{x_{i,2}} - {x_{j,2}}} \right)^2+\frac{{{\theta _2}}}{2}\left( {{x_{i,1}} - {x_{j,1}}} \right)^2
\end{align*}
for any $\theta_2 >0$, and 
\begin{align*}
&\frac{1}{2}\frac{d}{{dt}}{\left( {{x_{i,3}} - {x_{j,3}}} \right)^2} \nonumber\\
\le & ( {-a_3+\frac{{b_3^2}}{{2{\theta _3}}}} )\left( {{x_{i,3}} - {x_{j,3}}} \right)^2+\frac{{{\theta _3}}}{2}\left( {{x_{i,2}} - {x_{j,2}}} \right)^2
\end{align*}
for any $\theta_3 >0$. Combining the above inequalities, we obtain
\begin{align*}
&\frac{1}{2}\frac{d}{{dt}}{\left( {{x_{i,1}} - {x_{j,1}}} \right)^2}+\frac{1}{2}\frac{d}{{dt}}{\left( {{x_{i,2}} - {x_{j,2}}} \right)^2}+\frac{1}{2}\frac{d}{{dt}}{\left( {{x_{i,3}} - {x_{j,3}}} \right)^2} \nonumber\\
\le& ( { - {a_1} + {\theta } + \frac{{{\theta _1}}}{2}+\frac{{{\theta _2}}}{2}} )\left( {{x_{i,1}} - {x_{j,1}}} \right)^2+( {-a_2+\frac{{b_2^2}}{{2{\theta _2}}}+\frac{{{\theta _3}}}{2}} )\\
&\cdot \left( {{x_{i,2}} - {x_{j,2}}} \right)^2+( {-a_3+\frac{{b_3^2}}{{2{\theta _3}}}+\frac{{{\delta ^2}}}{{2{\theta _1}}}} )\left( {{x_{i,3}} - {x_{j,3}}} \right)^2\\
&+\left( {{x_{i,1}} - {x_{j,1}}} \right)\left( {u_i- u_j} \right)+\frac{{\tilde b_{i,j}^2}}{{2\theta }}\left({u_i^2+u_j^2}\right)\nonumber
\end{align*}
for any positive $\theta,\theta_1, \theta_2,\theta_3$. By choosing ${\theta _1} = \frac{{{\delta ^2}{\theta _3}}}{{2{a_3}{\theta _3} - b_3^2}}$ and ${\theta _2} = \frac{{b_2^2}}{{2{a_2} - {\theta _3}}}$, it follows that
$$
-a_2+\frac{{b_2^2}}{{2{\theta _2}}}+\frac{{{\theta _3}}}{2}=0 \text{ and }
-a_3+\frac{{b_3^2}}{{2{\theta _3}}}+\frac{{{\delta ^2}}}{{2{\theta _1}}}=0.
$$
Then we can rewrite the above inequality into 
\begin{align*}
&\frac{1}{2}\frac{d}{{dt}}{\left( {{x_{i,1}} - {x_{j,1}}} \right)^2}+\frac{1}{2}\frac{d}{{dt}}{\left( {{x_{i,2}} - {x_{j,2}}} \right)^2}+\frac{1}{2}\frac{d}{{dt}}{\left( {{x_{i,3}} - {x_{j,3}}} \right)^2} \nonumber\\
\le&( { - {a_1} + {\theta } + \frac{{{\theta _1}}}{2}+\frac{{{\theta _2}}}{2}})\left( {{x_{i,1}} - {x_{j,1}}} \right)^2\nonumber\\
&+\left( {{x_{i,1}} - {x_{j,1}}} \right)\left( {u_i- u_j} \right)+\frac{{\tilde b_{i,j}^2}}{{2\theta }}\left({u_i^2+u_j^2}\right).
\end{align*}
Integrating both sides from time $0$ to any $T>0$, we have
\begin{align*}
&\sum\limits_{q = 1}^3 {( {\frac{1}{2}{{\left( {{x_{i,q}}\left( T \right) - {x_{j,q}}\left( T \right)} \right)}^2} - \frac{1}{2}{{\left( {{x_{i,q}}\left( 0 \right) - {x_{j,q}}\left( 0 \right)} \right)}^2}} )} \\ 
\le &-\gamma_k\left\| {{x_{i,1}} - {x_{j,1}}} \right\|_T^2- \nu_k\left( {\left\| {{u_i}} \right\|_T^2 + \left\| {{u_j}} \right\|_T^2} \right)\\
&+ \left\langle {{u_i} - {u_j},{x_{i,1}} - {x_{j,1}}} \right\rangle_T ,
\end{align*}
which leads to
\begin{align*}
\left\langle {{u_i} - {u_j},{x_{i,1}} - {x_{j,1}}} \right\rangle_T  &\ge {\nu _k}\left( {\left\| {{u_i}} \right\|_T^2 + \left\| {{u_j}} \right\|_T^2} \right) \\
&\quad\quad\quad\quad+ {\gamma _k}\left\| {{x_{i,1}} - {x_{j,1}}} \right\|_T^2+\beta_k,
\end{align*}
where ${\nu _k} =  - \frac{{\tilde b_{i,j}^2}}{{2\theta }}$, ${\gamma _k} = { {a_1} - {\theta } - \frac{{{\theta _1}}}{2}-\frac{{{\theta _2}}}{2}} $, ${\beta _k} = \sum_{q = 1}^3 { - \frac{1}{2}{{\left( {{x_{i,q}}\left( 0 \right) - {x_{j,q}}\left( 0 \right)} \right)}^2}}$.
\end{proof}\par

Now we consider a network of five Goodwin oscillators $G_i,i\in\{1,2,3,4,5\}$ with parameters $a_1=0.5$, $a_2=a_3=1$, $b_2=b_3=1.5$, $p=14$. To introduce heterogeneity in the network, the input gains are chosen as $b_{1,1}=0.8$, $b_{2,1}=0.9$, $b_{3,1}=1$, $b_{4,1}=1.1$ and $b_{5,1}=1.2$. The oscillators are assumed to be interconnected through a complete graph, so that each oscillator is coupled with all the others. The coupling between oscillators is described by the function ${{\vartheta_k}\left( {x} \right)}=5x$, i.e., $\underline{\alpha}_k=\overline{\alpha}_k=5$, for all $e_k=(i,j)\in\mathcal{E}$. We can obtain from Proposition \ref{prop: case study 1} with $\theta=2$ and $\theta_3=1.5$ that Assumption~\ref{def: MI_OFP} holds with $\nu_k= -0.01$, and $\gamma_{k}=-16.1250$ for all $e_k=(i,j)\in\mathcal{E}$. Hence $\tilde\nu_i=-0.04,i\in\mathcal{N}$. Since the oscillators are interconnected through a complete graph, $r_i = 4, i\in\mathcal{N}$, and $\tilde r_k = 3$, which implied $\bar r_k = 0$ for all $e_k=(i,j)\in\mathcal{E}$. Therefore, it can be verified that $\frac{2+\tilde{r}_{k}}{\overline{\alpha}_{k}} -\frac{{  (1 + \underline{\alpha} _k^2){{\bar r}_k}}}{{2\underline{\alpha} _k^2}}+\frac{{{\gamma _{k}}}}{{{\underline{\alpha}_{k} ^2}}}-r_i |\tilde\nu_i| - r_j|\tilde\nu_j| = 0.035>0$ for all edges $e_k=(i, j) \in \mathcal{E}$. By Theorem \ref{thm: consensus condition}, the Goodwin oscillators network achieves IO synchronisation. To verify this, we set the initial value $x_{1,1}(0)=1.1$, $x_{2,1}(0)=-0.2$, $x_{3,1}(0)=1$, $x_{4,1}(0)=0.5$, $x_{5,1}(0)=0.3$ and the external disturbance to be $w_{ij}(t) = 0.3{\bar w_{ij}}(t)$, where ${\bar w_{ij}}(t)$ is white Gaussian noise with ${\bar w_{ij}}(t) \sim \mathcal{N}(0,1)$. As shown in Fig. \ref{fig.example1}
the outputs of the Goodwin oscillators reach synchronisation approximately with ${\left\| {D^{\top} Y} \right\|_T}$ bounded in terms of ${\left\| { W} \right\|_T}$.
\begin{figure}[!ht]
\centering
\includegraphics[width=8.3cm]{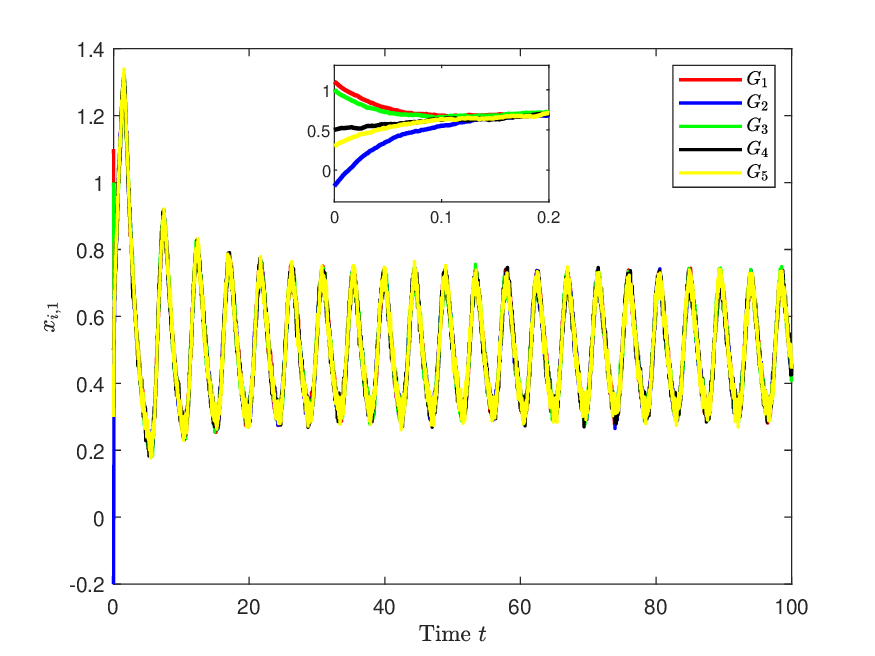}
\caption{Output trajectories of the Goodwin oscillators.}
\label{fig.example1}
\end{figure}

\section{Conclusion}\label{sec: Conclusion}
This work derived distributed conditions for input--output synchronisation in networks of heterogeneous nonlinear systems interconnected through nonlinear diffusive couplings. In particular, the proposed conditions rely solely on local information associated with adjacent agents and coupling links, and exploit relative dissipativity properties between neighbouring agents. 
The proposed approach establishes a network-level dissipativity property by reformulating pairwise dissipativity conditions among adjacent agents in terms of edge variables. This enables the analysis of the closed-loop interconnection via the dissipativity theorem. In contrast to existing centralised conditions, the resulting criteria can be verified in a fully distributed manner, making them scalable to large-scale networks. 
The effectiveness of the proposed framework was illustrated through a network of heterogeneous Goodwin oscillators, where the relative dissipativity properties were characterised and used to establish synchronisation. 



\bibliographystyle{IEEEtran}
\bibliography{reference}

\end{document}